\documentclass[letterpaper,10pt,conference]{ieeeconf}
\IEEEoverridecommandlockouts
\usepackage[compress]{cite}
\usepackage{amsmath,amssymb,amsfonts}
\usepackage{algorithmic}
\usepackage{graphicx}
\usepackage{textcomp}
\usepackage[dvipsnames]{xcolor}
\usepackage{comment}
\usepackage{soul}

\usepackage[hidelinks]{hyperref}
\usepackage{enumerate}

\usepackage{mathtools}
\newcommand{\defeq}{\coloneqq}

\usepackage[capitalize]{cleveref}

\usepackage[font=footnotesize]{caption}
\usepackage{subcaption}

\usepackage{multirow}
\usepackage{gensymb}
\usepackage{optidef}
\usepackage{booktabs}
\usepackage{commath}

\usepackage{soul}

\usepackage{pgfplots}
\usepgfplotslibrary{colormaps}
\usepgfplotslibrary{patchplots}
\pgfplotsset{compat=1.16}

\usepackage{matlab-prettifier}

\usepackage{cuted}
\usepackage{flushend}
\usepackage{accents}
\usepackage{lipsum}

\DeclareRobustCommand{\hlyellow}[1]{#1}
\DeclareRobustCommand{\hlorange}[1]{#1}
\DeclareRobustCommand{\hlgreen}[1]{#1}
\DeclareRobustCommand{\hlcyan}[1]{#1}

\newtheorem{theorem}{Theorem}

\newtheorem{lemma}{Lemma}

\newtheorem{remark}{Remark}
\newtheorem{problem}{Problem}

\newtheorem{proposition}{Proposition}

\title{\LARGE \bf
Stealthy Optimal Range-Sensor Placement for Target Localization
}

\author{Mohammad Hussein Yoosefian Nooshabadi, Rifat Sipahi, and Laurent Lessard%
\thanks{Research was sponsored by the Army Research Laboratory and was accomplished under Cooperative Agreement Number W911NF-23-2-0014. The views and conclusions contained in this document are those of the authors and should not be interpreted as representing the official policies, either expressed or implied, of the Army Research Laboratory or the U.S. Government. The U.S. Government is authorized to reproduce and distribute reprints for Government purposes notwithstanding any copyright notation herein.}%
\thanks{All authors are with the Department of Mechanical and Industrial Engineering, Northeastern University, Boston, MA 02115, USA.\newline
{\tt\footnotesize\{yoosefiannooshabad.m, r.sipahi, l.lessard\}\newline
@northeastern.edu}}
}

\begin{document}

\maketitle
\thispagestyle{empty}
\pagestyle{empty}

\begin{abstract}
We study a stealthy range-sensor placement problem where a set of range sensors are to be placed with respect to targets to effectively localize them while maintaining a degree of stealthiness from the targets. This is an open and challenging problem since two competing objectives must be balanced:  (a) optimally placing the sensors to maximize their ability to localize the targets and (b) minimizing the information the targets gather regarding the sensors. We provide analytical solutions in 2D for the case of any number of sensors that localize two targets.
\end{abstract}

\section{INTRODUCTION}\label{sec: intro}
We consider the problem of optimal sensor placement subject to stealthiness constraints. In this problem we have a network of range-only sensors and another network of stationary \emph{targets} (also equipped with range-only sensors). The goal is to obtain spatial configurations of the sensors that maximize their ability to localize the targets while limiting the targets' ability to localize the sensors.

This problem concerns two competing objectives, \hlorange{each of which has been studied extensively.} The first competing objective is \emph{target localization} where the goal is to solely optimize the sensors' localization performance \cite{bishop2010optimality, moreno2013optimal, sadeghi2020optimal}. In \cite{moreno2013optimal} the optimal relative sensor-target configuration is derived for multiple sensors and multiple targets in 2D settings. In \cite{sadeghi2020optimal} the scenario with multiple sensors and a single target is considered where the sensors are constrained to lie inside a connected region. Both aforementioned works characterize localization performance using the Fisher Information Matrix (FIM)
\cite[\S2]{barfoot2024state}. Target localization is a special case of \textit{optimal sensor placement} where the goal is to find the optimal location of a network of sensors such that some notion of information they obtain is maximized; \hlyellow{this problem has many applications, including structural health monitoring} \cite{ostachowicz2019optimization} \hlyellow{and experiment design} \cite{zayats2010optimal}. \looseness=-1

The second competing objective is \textit{stealthy sensor placement} where the goal is to place the sensors in spatial configurations such that the localization performance of the targets is limited. Various measures of localization performance for the adversary sensors are introduced in the literature, including entropy \cite{molloy2023smoother}, predictability exponent \cite{xu2022predictability} and FIM \cite{farokhi2020privacy}. Stealthiness has also been studied in the context of mobile sensors for \hlyellow{different applications} such as adversarial search-and-rescue \cite{rahman2022adversar}, information acquisition \cite{schlotfeldt2018adversarial}, pursuit-evasion \cite{chung2011search} and covert surveillance \cite{huang2022decentralized}.
The work in \cite{karabag2019least} uses the FIM to make the policy of a single mobile sensor difficult to infer for an adversary Bayesian estimator. In \cite{khojasteh2022location} a sensor equipped with a range sensor is allowed to deviate from its prescribed trajectory to enhance its location secrecy encoded via the FIM. The notion of stealthiness is also related to the topics of \textit{privacy} and \textit{security}, which have \hlyellow{applications in} numerical optimization \cite{farokhi2020privacy}, machine learning \cite{zhang2016dynamic}, \hlyellow{smart grids} \cite{li2018information}, and communication \cite{wang2021physical}.

In the present paper, \hlorange{we combine the two aforementioned objectives by optimizing localization performance subject to a stealthiness constraint, quantifying each using a min-max formulation of the FIM. To the best of our knowledge, the present study is the first to consider this combination of objectives.} \hlyellow{Possible applications include cooperative distributed sensing and cyber-physical system security; detecting suspicious activities without alerting the attackers.} 

The paper is organized as follows. In \cref{sec: problem setup}, we describe our min-max FIM formulation. The general case of arbitrarily many sensors and targets leads to a difficult non-convex problem, thus we focus on more tractable special cases. In \cref{sec: Problem formulation 2t2a} we provide a complete solution in the case of two sensors and two targets in 2D. Next, in \cref{sec:Extension}, we treat the case of arbitrarily many sensors and two targets and provide various analytic performance bounds. In \cref{sec: Conclusion and Future Work}, we conclude and discuss future directions.

\section{PROBLEM SETUP} \label{sec: problem setup}

Consider a 2D arrangement of $m$ sensors $s_1,\dots,s_m$ and $n$ targets $t_1,\dots,t_n$. A possible configuration of these sensors and targets is depicted in  \cref{fig: problemDefinition}. We let $\theta_{ij,k}$ denote the angle between $s_i$ and $s_j$ as viewed from $t_k$. Similarly, we let $\beta_{k\ell,j}$ denote the angle between $t_k$ and $t_\ell$ as viewed from $s_j$. 

\begin{figure}[b!]
    \centering
    \includegraphics{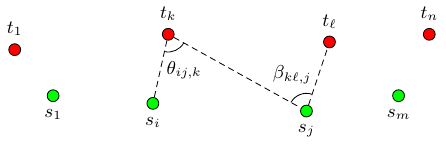}
    \caption{The problem setup in this paper. A set of $m$ sensors (in green) are to be placed such that their ability to localize a set of $n$ targets (in red) is maximized while the targets' ability to localize the sensors is limited.}
    \label{fig: problemDefinition}
\end{figure}

We assume sensors and targets use \emph{range-only sensing} and each \hlorange{measurement is subject to additive zero-mean Gaussian noise with variance $\sigma^2$.} Due to the noise, the spatial configuration of the sensors relative to the targets is critical to effectively fuse sensor measurements for localization purposes. The FIM is broadly used to quantify the quality of localization \cite{martinez2006optimal}; it is the inverse of the covariance matrix of the target's position conditioned on the observed measurements. We use the so-called \emph{D-optimality} criterion, which is the determinant of the FIM. 
Other FIM-based criteria include the A-optimality or E-optimality \cite[\S7]{boyd2004convex}. Since we are using range-only sensing in a 2D setting, the D-optimality, A-optimality, and E-optimality criteria are equivalent \cite{sadeghi2020optimal}. 

We denote by $\mathcal{I}_k$ the D-optimality criterion for sensors $s_1,\dots,s_m$ localizing a target $t_k$ using their collective range measurements. Similarly we denote by $\mathcal{J}_i$ the D-optimality criterion for targets $t_1,\dots,t_n$ localizing a sensor $s_i$, where
\begin{align}\label{eq: det FIM of target k}
    \mathcal{I}_k & = \frac{1}{\sigma^4}\sum_{1 \leq i < j \leq m}\mkern-18mu\sin^2{\theta_{ij, k}}, &
    \mathcal{J}_i & = \frac{1}{\sigma^4}\sum_{1\leq k < \ell \leq n}\mkern-18mu\sin^2{\beta_{k\ell, i}}.
\end{align}
For a detailed derivation of \eqref{eq: det FIM of target k}, see \cite{martinez2006optimal}.
\hlgreen{Assuming Gaussian noise is critical in obtaining {\eqref{eq: det FIM of target k}}. Intuitively, information from two range measurements is maximized when the range vectors are perpendicular and minimized when they are parallel.} 

Motivated by the goal of obtaining the best possible localization of a set of targets while avoiding detection by those same targets, we formulate our problem as follows.

\begin{problem}[\textit{min-max optimal stealthy sensor placement}]\label{problem 1}
Given target locations $t_1,\dots,t_n$ find angles $\theta_{ij, k}$ and $\beta_{k\ell, i}$ shown in \cref{fig: problemDefinition} corresponding to a feasible arrangement of the sensors $s_1,\dots,s_m$ such that the minimum information that the sensors obtain about the targets is maximized while the maximum information that the targets obtain about the sensors is less than some prescribed level $\gamma^2$. We call $\gamma$ the \textit{information leakage level}.
\end{problem}

We can formulate \cref{problem 1} as the  optimization problem
\begin{align}\label{eq: problem def general}
        \underset{\theta,\, \beta}{\text{maximize}} \quad & \min_{1 \leq k \leq n} \mathcal{I}_k\\
        \text{subject to:} \quad & \max_{1\leq i \leq m} \mathcal{J}_i \leq \gamma^2 \notag\\
       & (\theta, \beta) \in \mathcal{F}, \notag
\end{align}
\hlgreen{where $\mathcal{F}$ is the set of all geometrically feasible $(\theta,\beta)$.
This constraint ensures that the spatial arrangement of sensors and targets with angles $\theta$ and $\beta$ is realizable. We compute $\mathcal{F}$ for the special case $m=n=2$ in} \cref{prop:cases}.

\begin{remark}
Instead of constraining the maximum of $\mathcal{J}_i$ (the most information the targets have about any particular sensor), one could constrain the sum or the product of $\mathcal{J}_i$ or some other norm of the vector $(\mathcal{J}_1,\dots,\mathcal{J}_m)$. It is also possible to apply different norms to the expression of $\mathcal{I}_k$.
\end{remark}

In the next section, we solve \cref{problem 1} for the special case of $m=2$ sensors and $n=2$ targets.

\section{TWO SENSORS AND TWO TARGETS}\label{sec: Problem formulation 2t2a}
\hlcyan{Substituting {\eqref{eq: det FIM of target k}} into {\eqref{eq: problem def general}} with $m=n=2$ yields}
\begin{subequations}\label{eq: problem def for 2s2t}
\begin{align}
        \underset{\theta_1, \theta_2, \beta_1, \beta_2}{\text{maximize}} \quad & \min \bigl(\sin^2\theta_1, \, \sin^2\theta_2\bigr)  \label{opt:obj}\\ 
        \text{subject to:} \quad & \max\bigl(\sin^2\beta_1, \, \sin^2\beta_2\bigr) \leq \gamma^2 \label{opt:gamma}\\ 
        & (\theta_1, \theta_2, \beta_1, \beta_2) \in \mathcal{F}, \label{opt:h}
\end{align}
\end{subequations}
where we have used the simplified notation $\theta_k \defeq \theta_{12,k}$ and similarly for $\beta_k$. \hlorange{We also assumed without loss of generality that $\sigma = 1$.} We analyze \eqref{eq: problem def for 2s2t} in three steps:

\paragraph{The objective \eqref{opt:obj}} In epigraph form, the level sets $\min\bigl(\sin^2\theta_1, \, \sin^2\theta_2\bigr) \geq \eta^2$ consist of configurations where $\sin\theta_k \geq \eta$ for $k=1,2$. In other words,
\begin{equation}\label{eq:thetarange}
\theta_k \in [ \arcsin\eta, \pi-\arcsin\eta ]
\qquad\text{for }k=1,2.
\end{equation}
The set of points $t_k$ for which $\theta_k$ is fixed is the arc of a circle passing through $s_1$ and $s_2$. Therefore, given $\eta$, sublevel sets for all $\theta_k$ in \eqref{eq:thetarange} are the exclusive-OR between two congruent discs whose boundaries intersect at $s_1$ and $s_2$ (as shown in \cref{fig: level sets} Left). The objective is maximized at the value 1 when $\theta_1=\theta_2=\tfrac{\pi}{2}$. This corresponds to the configuration when both circles coincide with the dotted circle and $t_k$ lies on this circle. This is logical since localization error is minimized when the range vectors are orthogonal.

\begin{figure}[ht]
\centering
\includegraphics[page=1]{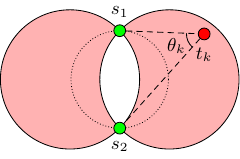}
\includegraphics[page=2]{fig_level_sets}
\caption{\textbf{Left:} Shaded region that must contain $t_1$ and $t_2$ (relative to $s_1$ and $s_2$) so that the objective \eqref{opt:obj} is at least $\eta^2$. The region is formed by two intersecting circles ($\eta=0.7$ shown). The dotted circle shows $\eta=1$.
\textbf{Right:} Shaded region that must contain $s_1$ and $s_2$ (relative to $t_1$ and $t_2$) so that information leakage level is at most $\gamma$ ($\gamma=0.7$ shown).}
\label{fig: level sets}
\end{figure}

\paragraph{Stealthiness constraint \eqref{opt:gamma}} Similar to \eqref{eq:thetarange} we have
\begin{equation}\label{eq:betarange}
\beta_i \in [0,\arcsin\gamma] \cup [\pi-\arcsin\gamma, \pi]
\quad\text{for }i=1,2.
\end{equation}
The set of admissible $\beta_i$ is shown in \cref{fig: level sets} Right. Intuitively, the targets have a large localization error for a sensor when the targets' range vectors are close to being parallel ($\beta_i$ is close to $0$ or to $\pi$). This splits the feasible set into two disjoint regions: \emph{between} $t_1$ and $t_2$ or \emph{outside}.

\paragraph{Feasible configurations \eqref{opt:h}}  Considering a quadrilateral whose vertices are formed by the sensors and the targets, there are several cases to consider; depending on whether any interior angles are greater than $\pi$ or whether the quadrilateral is self-intersecting. By inspection we have seven distinctive cases, illustrated in \cref{fig: seven cases}. Each case corresponds to a set of constraints given next in \cref{prop:cases}.

\begin{figure*}
\vspace{2mm}
    \centering
    \includegraphics[page=1]{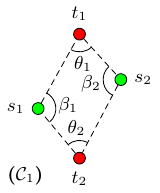}\hfill
    \includegraphics[page=2]{fig_seven_cases}\hfill
    \includegraphics[page=3]{fig_seven_cases}\hfill
    \includegraphics[page=4]{fig_seven_cases}\hfill
    \includegraphics[page=5]{fig_seven_cases}\hfill
    \includegraphics[page=6]{fig_seven_cases}\hfill
    \includegraphics[page=7]{fig_seven_cases}
    \caption{
    The seven possible configurations for two sensors and two targets. Each case is characterized by different constraints relating $\theta_1$, $\theta_2$, $\beta_1$ and $\beta_2$ which are given in \cref{prop:cases}.
    If sensors $s_1$ and $s_2$ are interchangeable then $\mathcal{C}_4 = \mathcal{C}_5$ and $\mathcal{C}_6 = \mathcal{C}_7$ and there are only five distinct cases.
    }
    \vspace{-2mm}
    \label{fig: seven cases}
\end{figure*}

\begin{proposition}\label{prop:cases}
    The feasible configurations $\mathcal{F}$ in \eqref{opt:h} is the union of the seven cases shown in \cref{fig: seven cases}. In other words,
    \begin{equation*}
        \mathcal{F} = \biggl\{(\theta_1, \theta_2, \beta_1, \beta_2) \in [0,\pi]^4
        \;\bigg|\; \bigcup_{i=1}^7 \mathcal{C}_i\biggr\},
    \end{equation*}
    where $\mathcal{C}_i$ for $i=1,\dots,7$ is the constraint set corresponding to the $i\textsuperscript{th}$ case shown in \cref{fig: seven cases} and expressed below:
\begin{subequations}
    \begin{align*}
        \mathcal{C}_1 &: 
            &&\theta_1 + \theta_2 + \beta_1 + \beta_2 = 2\pi\\ 
        \mathcal{C}_2 &: 
             &-&\theta_1 + \theta_2 + \beta_1 + \beta_2 = 0,\;\;
            \theta_2 \leq \theta_1\\
        \mathcal{C}_3 &:  
            &&\theta_1 - \theta_2 + \beta_1 + \beta_2 = 0,\;\;
            \theta_1 \leq \theta_2\\
        \mathcal{C}_4 &:  
            &&\theta_1 + \theta_2 - \beta_1 + \beta_2 = 0,\;\;
            \beta_2 \leq \beta_1\\
        \mathcal{C}_5 &: 
            &&\theta_1 + \theta_2 + \beta_1 - \beta_2 = 0,\;\;
            \beta_1 \leq \beta_2\\
        \mathcal{C}_6 &:
            &&\theta_1 - \theta_2 + \beta_1 - \beta_2 = 0,\;\;
            \theta_1 + \beta_1 \leq \pi \\
        \mathcal{C}_7 &:
            &&\theta_1 - \theta_2 - \beta_1 + \beta_2 = 0,\;\;
        \theta_1 + \beta_2 \leq \pi 
    \end{align*}
\end{subequations}
If the sensors $s_1$ and $s_2$ are interchangeable (e.g., if there are no additional constraints that distinguish $s_1$ from $s_2$) then swapping $\beta_1$ and $\beta_2$ leads to $\mathcal{C}_4 = \mathcal{C}_5$ and $\mathcal{C}_6 = \mathcal{C}_7$.
\end{proposition}

\medskip

The following result is useful in the sequel and is a straightforward consequence of the constraint equations $\mathcal{C}_i$.

\medskip

\begin{lemma}\label{lem:C1}
Suppose $(\theta_1, \theta_2, \beta_1, \beta_2) \in \mathcal{F}$, where $\mathcal{F}$ is defined in \cref{prop:cases}. Then $\theta_1 + \theta_2 + \beta_1 + \beta_2 \leq 2\pi$, where equality is achievable in a non-degenerate configuration (no sensor is placed arbitrarily close to a target) if and only if $(\theta_1, \theta_2, \beta_1, \beta_2) \in \mathcal{C}_1$.
\end{lemma}

\medskip

We provide the analytical results in two theorems. Our first theorem considers the unconstrained case where the sensors may be placed anywhere in relation to the targets.
\begin{theorem}\label{theorem 1}
  Consider the optimization problem \eqref{eq: problem def for 2s2t}. If the sensors $s_1$ and $s_2$ can be freely placed anywhere then a configuration of sensors is globally optimal if and only if:
  \begin{enumerate}[(i)]
  \item $s_1$, $s_2$, $t_1$, $t_2$ are \emph{cyclic} (they lie on a common circle),
  \item $s_1$ and $s_2$ are diametrically opposed on this circle, and
  \item The common circle has diameter at least $d/\gamma$; where $d$ is the distance between $t_1$ and $t_2$.
  \end{enumerate}
  Moreover, any such configuration satisfies $\theta_1 = \theta_2 = \frac{\pi}{2}$ and $\sin\beta_1 = \sin\beta_2$ and has an optimal objective value of $1$.
\end{theorem}

\begin{proof}
    The objective \eqref{opt:obj} is upper-bounded by $1$, which is achieved if and only if $\theta_1=\theta_2=\frac{\pi}{2}$. This is possible if and only if $s_1$, $s_2$, $t_1$ and $t_2$ are on a common circle with diameter $\abs{s_1s_2}$. 
    If $s_1$ and $s_2$ lie on alternating sides of $t_1$ and $t_2$, we recover $\mathcal{C}_1$ from \cref{fig: seven cases}. Otherwise, we recover cases $\mathcal{C}_6$ or $\mathcal{C}_7$.
    Given such a configuration where the common circle has diameter $D$ apply the Law of Sines to $\triangle t_1 t_2 s_i$ and obtain $d/\sin\beta_i = D$. Now \eqref{opt:gamma} is equivalent to $\sin\beta_i \leq \gamma$ and therefore $D \geq d/\gamma$.
\end{proof}

Examples of optimal configurations proved in Theorem \ref{theorem 1} for the cases $\mathcal{C}_1$ and $\mathcal{C}_6$ or $\mathcal{C}_7$ are illustrated in \cref{fig: 2t2s Theorem1}.

\begin{figure}[ht]
    \centering
    \includegraphics[page=1]{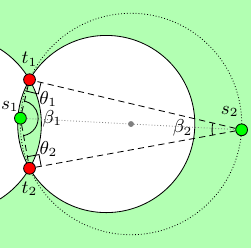}
    \includegraphics[page=2]{fig_thm1}
    \caption{Examples of optimal sensor configurations when positions are unconstrained (\cref{theorem 1}). The solid circles delineate the feasible set (see \cref{fig: level sets}). The dotted circle is any larger circle passing through $t_1$ and $t_2$. A configuration is optimal if and only if the sensors $s_1$ and $s_2$ lie on this larger circle and are diametrically opposed. This can happen with alternating sensors and targets (left) or with both sensors on the same side (right).}
    \label{fig: 2t2s Theorem1}
\end{figure}

Recall that the feasible set is split into two disjoint regions (see \cref{fig: level sets} Right). \cref{theorem 1} confirms that optimal configurations exist with one sensor between $t_1$ and $t_2$ and one outside or with both sensors outside (see \cref{fig: 2t2s Theorem1}). We next investigate the scenarios where \emph{both} sensors are between $t_1$ and $t_2$; that is $\beta_i \in [\pi-\arcsin\gamma,\pi]$ for $i=1,2$.

\begin{theorem}\label{theorem 2}
    Consider \eqref{eq: problem def for 2s2t} with the extra constraint $\beta_i \in [\pi-\arcsin\gamma,\pi]$ for $i=1,2$.
    Then, a non-degenerate configuration of sensors is globally optimal if and only if:
  \begin{enumerate}[(i)]
  \item $s_1$, $t_1$, $s_2$ and $t_2$ are vertices of a parallelogram, and
  \item $s_1$ and $s_2$ are on different circles through $t_1$ and $t_2$ with diameter $d/\gamma$; where $d$ is the distance between $t_1$, $t_2$.
  \end{enumerate}
Any such configuration satisfies $\theta_1=\theta_2 = \arcsin\gamma$ and $\beta_1=\beta_2 = \pi-\arcsin\gamma$ and has an optimal objective value of $\gamma^2$.
\end{theorem}

\begin{proof}
    We prove sufficiency first.
    Incorporating the constraints on $\beta_1$ and $\beta_2$ and using the fact that $\sin(\cdot)$ is nonnegative on $[0,\pi]$, we can rewrite \eqref{eq: problem def for 2s2t} as
    \begin{align}\label{eq: problem def for 2s2t simplified}
            \underset{\theta_1, \theta_2, \beta_1, \beta_2}{\text{maximize}} \quad & \min (\sin\theta_1, \, \sin\theta_2)\\
            \textrm{subject to:} \quad &\pi-\arcsin\gamma \leq \beta_i \leq \pi & i&\in\{1,2\}\notag\\
            & 0 \leq \theta_k \leq \pi & k&\in\{1,2\}\notag\\
            & (\theta_1, \theta_2, \beta_1, \beta_2) \in \mathcal{F}.\notag
    \end{align}
    By concavity of $\sin(\cdot)$ on $[0,\pi]$ and Jensen's inequality,
    \begin{equation}\label{ineq1}
        \min(\sin\theta_1, \, \sin\theta_2)
        \leq \frac{\sin\theta_1 \!+\! \sin\theta_2}{2}
        \leq \sin\biggl( \frac{\theta_1+\theta_2}{2} \biggr).
    \end{equation}
    From \cref{lem:C1} we have $\theta_1 + \theta_2 + \beta_1 + \beta_2 \leq 2\pi$ with equality only in case $\mathcal{C}_1$. Therefore,
    \begin{equation}\label{ineq2}
    0 \leq \frac{\theta_1+\theta_2}{2} \leq \frac{2\pi-\beta_1-\beta_2}{2} \leq \arcsin\gamma \leq \frac{\pi}{2}.
    \end{equation}
    Since $\sin(\cdot)$ is monotonically increasing on $[0,\tfrac{\pi}{2}]$ we may combine \eqref{ineq1} and \eqref{ineq2} to upper-bound the objective of \eqref{eq: problem def for 2s2t simplified}
    \[
    \min(\sin\theta_1, \, \sin\theta_2)
    \leq \sin\biggl( \frac{\theta_1+\theta_2}{2} \biggr)
    \leq \gamma,
    \]
    with equality only possible for case $\mathcal{C}_1$. Indeed setting $\theta_1=\theta_2=\arcsin\gamma$ and $\beta_1=\beta_2=\pi-\arcsin\gamma$ renders the bound tight and satisfies the constraint for case $\mathcal{C}_1$ in \cref{prop:cases} and is therefore optimal. This solution is illustrated in \cref{fig: 2t2s Theorem2}.

    \begin{figure}[ht]
    \centering
    \includegraphics{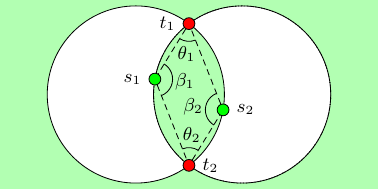}
    \caption{Example of an optimal sensor configuration when sensors are constrained to lie in the region between both targets (\cref{theorem 2}). Solid circles delineate the feasible set (see \cref{fig: level sets}). A configuration is optimal if  $s_1$ and $s_2$ lie on opposite arcs and $t_1$, $s_1$, $t_2$ and $s_2$ form a parallelogram.}
    \label{fig: 2t2s Theorem2}
\end{figure}

    Necessity follows from convexity of \eqref{eq: problem def for 2s2t simplified}. Indeed, from the derivation above, any optimal non-degenerate configuration must be of case $\mathcal{C}_1$ (see \cref{fig: seven cases}), so constraint $\mathcal{C}_1$ holds (all constraints are linear). Moreover, the objective of \eqref{eq: problem def for 2s2t simplified} is the minimum of two concave functions, so it is concave.
\end{proof}

\begin{remark}
    \hlcyan{Letting $s_1$ and $s_2$ approach $t_1$ and $t_2$, respectively, yields in the limit $\beta_1 = \beta_2 = \pi - \arcsin\gamma$ and $\theta_1 = \theta_2 = \arcsin\gamma$, which is a globally optimal (but degenerate) configuration of cases $\mathcal{C}_1$, $\mathcal{C}_6$, or $\mathcal{C}_7$.}
\end{remark}

\section{ARBITRARILY MANY SENSORS}\label{sec:Extension}

In this section we extend our results to the case with $n=2$ targets and \hlorange{$m\geq 3$} sensors. When the sensors can be freely placed (as considered in \cref{theorem 1}) a globally optimal solution is to place the sensors infinitely far from the targets in a circular arrangement. Under these conditions the stealthiness constraint is automatically satisfied and the problem reduces to optimal sensor placement  with a single target \cite{moreno2013optimal, sadeghi2020optimal}. 

\begin{figure}[ht]
    \vspace{1mm}
    \centering
    \includegraphics{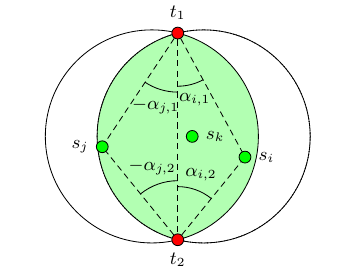}
    \caption{The problem with $n=2$ targets and $m \geq 3$ sensors, where sensors are constrained to be inside the green region ($\gamma = 0.95$ shown). Sensor $s_i$'s position is characterized by heading angles $\alpha_{i,1}$ and $\alpha_{i,2}$ from targets $t_1$ and $t_2$, respectively. Sensors have positive heading angles when right of the center-line and negative when left.}
    \vspace{-1mm}
    \label{fig: ms2t}
\end{figure}

The more practical, yet challenging, scenario is when the sensors are constrained to lie in the region between $t_1$ and $t_2$ (as considered in \cref{theorem 2}). It is convenient to re-parameterize this problem in terms of the \emph{heading angles} $\alpha_{i,k}$ of sensor $s_i$ viewed from target $t_k$ and measured relative to the center-line joining $t_1$ and $t_2$ (as illustrated in \cref{fig: ms2t}).
By convention $\alpha_{i,k}>0$ if $s_i$ is to the right of the center-line and $\alpha_{i,k}<0$ if $s_i$ is to the left. Thus we may set $\theta_{ij,k} = |\alpha_{i,k}-\alpha_{j,k}|$ and $\beta_{i} = \pi-|\alpha_{i,1}+\alpha_{i,2}|$ and using \eqref{eq:betarange}, write \eqref{eq: problem def general} in epigraph form as
\begin{align}\label{eq: problem def for ms2t}
            \underset{\alpha, \eta^2}{\text{maximize}} \quad & \eta^2\\
            \textrm{subject to:} \quad 
            & \mkern-18mu\sum_{1 \leq i < j \leq m}\mkern-18mu\sin^2(\alpha_{i,k}-\alpha_{j,k}) \geq \eta^2 & k&\in\{1,2\}\notag\\
            & \abs{\alpha_{i, 1} + \alpha_{i, 2}} \leq \arcsin \gamma & i&\in\{1,\dots,m\}\notag\\
            & \alpha_{i,1}\alpha_{i,2} \geq 0 & i&\in\{1,\dots,m\}.\notag
\end{align}
The final constraint ensures that $\alpha_{i,1}$ and $\alpha_{i,2}$ have the same sign, so that the $i\textsuperscript{th}$ rays emanating from $t_1$ and $t_2$ are guaranteed to intersect at $s_i$.

\cref{eq: problem def for ms2t} is non-convex due to the first and last constraints. A closed-form solution remains elusive, so in the sections that follow, we turn our attention to deriving analytic upper and lower bounds to the optimal $\eta^2$.

\subsection{Analytic lower bounds}\label{subsec: lower bound analytical}

For the maximization problem \eqref{eq: problem def for ms2t}, any feasible configuration of sensors yields a lower bound on the optimal objective value. Based on the results in \cref{sec: Problem formulation 2t2a}, we focus on two important configurations with associated lower bound
\begin{equation}\label{eq: lower bound symmetry}
    \eta^2
    \;
    = \sum_{1 \leq i < j \leq m}\mkern-18mu\sin^2(\alpha_{i,k}-\alpha_{j,k})\quad\text{for }k\in\{1,2\}.
\end{equation}

\paragraph{Degenerate configuration}\label{section:degenerate_configuration}
Refer to \cref{fig: dgnrt and unfrm} Left and consider the two arcs that form the boundaries of the green region. First, place two sensors on the left arc arbitrarily close to $t_1$ and $t_2$, respectively ($s_1$ and $s_2$). Then similarly place two sensors on the right arc ($s_3$ and $s_4$). Continue placing two sensors at a time on alternating arcs until all sensors are placed. If there is an odd number of sensors then place the last sensor in the center of the arc ($s_5$).

Since the sensors are either arbitrarily close to the targets or in the middle of one of the arcs, the angles $\alpha_{i, k}$ are equal to $\pm\psi$ or $\pm\phi$ (see \cref{fig: dgnrt and unfrm}, Left). By straightforward calculations
we obtain $\psi = 2\phi = \arcsin\gamma$. As a result, using \eqref{eq: lower bound symmetry}, we can compute the lower bound analytically
\begin{equation}\label{eq: analytical lb degenerate}
    \eta^2_\textup{opt} \geq
    \left\{
    \begin{aligned}
    &\tfrac{1}{4} m^2 \gamma^2 (2-\gamma^2) \\
    &\tfrac{1}{4} (m\!-\!1) \bigl(2-2 (1\!-\!\gamma ^2)^{3/2}+(m\!-\!1)\gamma ^2 (2\!-\!\gamma ^2)\bigr) \\
    &\tfrac{1}{4} m^2 \gamma ^2 (2-\gamma ^2) -  \gamma ^2 + \gamma^4 \\
    &\tfrac{1}{4} (m\!-\!1) \bigl(2-2 (1\!-\!\gamma ^2)^{3/2}+(m\!-\!1)\gamma ^2 (2\!-\!\gamma ^2)\bigr) \\
    &\hspace{3.2cm} -\gamma^2+\gamma^4+\gamma ^2\sqrt{1-\gamma ^2}\hspace{-0.3cm}
    \end{aligned}
    \right.
\end{equation}

\noindent where the four lower bounds in \eqref{eq: analytical lb degenerate} correspond to the cases $m\equiv 0,1,2,3 \pmod{4}$, respectively.

\begin{figure}[t]
\centering
\includegraphics[page=1]{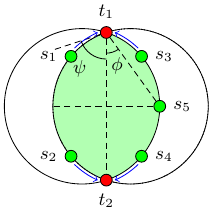}
\includegraphics[page=2]{fig_dgnrt_unfrm}
\caption{Two configurations used to lower-bound \eqref{eq: problem def for ms2t}. \textbf{Left:} \emph{Degenerate} configuration. Sensors are placed on the boundary of the green region (arbitrarily close to the targets) forming an angle $\psi$ with the center line. When $m$ is odd the last sensor is placed in the middle of the arc with fewest sensors on it; forming an angle $\phi = \frac{1}{2}\psi$ with the center line. \textbf{Right:} \emph{Uniform} configuration. Sensors are placed such that the lengths of the arcs between neighboring sensors (shown in brown) are equal.}
\label{fig: dgnrt and unfrm}
\end{figure}

\paragraph{Uniform configuration} 
In this configuration, the sensors are placed on the boundaries of the green region such that the arc length between neighboring sensors is equal to $2\delta$ (see \cref{fig: dgnrt and unfrm} Right). Straightforward computation yields
$\delta = \frac{2}{m}\arcsin\gamma$.
It can be seen from \cref{fig: dgnrt and unfrm} (Right) that the $(\alpha_{i,k}-\alpha_{j,k})$ are multiples of $\delta$. Direct calculation gives
\begin{multline}\label{eq: analytical lb uniform}
    \eta_\textup{opt}^2 \geq 
    \sum_{1 \leq i < j \leq m}\mkern-18mu\sin^2(\alpha_{i,1}-\alpha_{j,1})
    = \mkern-18mu\sum_{1 \leq k \leq m-1}\mkern-18mu(m-k)\sin^2(k\delta)\\
    = \frac{1}{8\sin^2\delta}\left( m^2 - 1 + \cos(2m\delta) - m^2\cos(2\delta) \right),
\end{multline}
\cref{eq: analytical lb uniform} together with $\delta = \frac{2}{m}\arcsin\gamma$ gives the second analytical lower bound to \eqref{eq: problem def for ms2t}.

\subsection{Analytic upper bound}\label{subsec: upper bound analytical}
We provide two upper-bounds for \eqref{eq: problem def for ms2t}.

\paragraph{Constraint approach}
\hlorange{Let $\mathcal{S}$ and $\mathcal{D}$ denote the pairs of indices corresponding to agents on the \emph{same} side and \emph{different} side of the centerline between both targets.}
\hlorange{Since $|\alpha_{i,k}|\leq\arcsin\gamma$ we can upper-bound $|\alpha_{i,k}-\alpha_{j,k}|$ based on whether $(i,j)$ are on the same side (in $\mathcal{S}$) or on different sides (in $\mathcal{D}$).}
As a result, we have for $k\in\{1,2\}$
\begin{equation*}
    \sin^2(\alpha_{i,k}-\alpha_{j,k}) \leq
    \begin{cases}
        \gamma^2 & (i,j) \in \mathcal{S}\\
        4\gamma^2(1-\gamma^2) & (i,j) \in \mathcal{D},\;\gamma \leq \frac{1}{\sqrt{2}} \\
        1 & (i,j) \in \mathcal{D},\;\gamma > \frac{1}{\sqrt{2}}.
    \end{cases}
\end{equation*}
Assuming we have $p$ agents on the right side and $m-p$ on the left, we may apply the bound above and obtain:
\begin{align}\label{eq: upper bound p dependent}
    \eta^2 &= \sum_{1\leq i < j \leq m} \mkern-18mu\sin^2(\alpha_{i,k}-\alpha_{j,k})
    \leq
        \biggl(\binom{p}{2} + \binom{m-p}{2}\biggr)\gamma^2 \notag\\
        &\hspace{1cm} + m(m-p) \begin{cases}
            4\gamma^2(1-\gamma^2) & \gamma \leq \frac{1}{\sqrt{2}} \\
            1 & \gamma > \frac{1}{\sqrt{2}}.
        \end{cases}
\end{align}
Maximizing the right-hand side of \eqref{eq: upper bound p dependent} over $p$ gives $p^\star = \frac{m}{2}$ for even and $p^\star = \frac{m-1}{2}$ for odd values of $m$ and
\begin{equation}\label{eq: upper bound first method}
    \eta^2_{\text{opt}} \leq \begin{cases}
        \frac{1}{4}m\gamma^2\bigl(m(5 -4\gamma^2) - 2\bigr) & \gamma \leq \frac{1}{\sqrt{2}}\\
        \frac{1}{4}m\bigl(m - 2\gamma^2 + m \gamma^2 \bigr) & \gamma > \frac{1}{\sqrt{2}}.
    \end{cases}
\end{equation}
\cref{eq: upper bound first method} is an analytical upper bound to \eqref{eq: problem def for ms2t}.

\paragraph{Jensen approach}
Let $g(\theta)$ be a \hlorange{\emph{concave non-decreasing}} upper bound $\sin^2\theta \leq g(\theta)$. Applying Jensen's inequality to \eqref{eq: lower bound symmetry}, we obtain for $k\in\{1,2\}$
\begin{align}\label{eta2 upper bound jensen}
    \eta^2 &= \sum_{1 \leq i < j \leq m}\mkern-10mu\sin^2(\alpha_{i,k}-\alpha_{j,k}) \notag \\
    &\leq \binom{m}{2}g\Biggl(\binom{m}{2}^{-1}\mkern-18mu\sum_{1 \leq i < j \leq m}\mkern-10mu (\alpha_{i,k}-\alpha_{j,k})\Biggr).
\end{align}
The upper bound is determined by choosing $\alpha_{i,k}$ in order to maximize the minimum of the right-hand side of \eqref{eta2 upper bound jensen} for $k\in\{1,2\}$. We next sketch the derivation. Start by assuming without loss of generality that $\alpha_{1, 1} \geq \dots \geq \alpha_{m, 1}$, then
\begin{equation}\label{eq: sum theta 1}
    \sum_{1 \leq i < j \leq m} \mkern-18mu (\alpha_{i,1}-\alpha_{j,1}) = \sum_{i=1}^m (m - 2i + 1)\alpha_{i, 1}.
\end{equation}
The steps involve showing that:
(1) half of the $\alpha_{i,1}$ must be nonegative and the other half nonpositive,
(2) The second constraint of \eqref{eq: problem def for ms2t} is tight leading to $\alpha_{i,2}+\alpha_{i,1} = \pm\arcsin\gamma$ and $\alpha_{1,2}\leq \cdots \leq \alpha_{m,2}$, and
(3) The positive-coefficient and negative-coefficient terms of \eqref{eq: sum theta 1} can be separately upper-bounded using $\min(a,b) \leq \tfrac{1}{2}(a+b)$, which has a trivial solution of letting half the $\alpha_{i,k}$ equal to 0 and the other half equal to $\pm \arcsin\gamma$.
This solution leads to
\begin{equation}\label{eq: upper bound second method}
    \eta^2_{\text{opt}} \leq \begin{cases}
        \binom{m}{2}g\Bigl( \frac{3m}{4(m-1)}\arcsin\gamma \Bigr) & m\text{ even}\\
        \binom{m}{2}g\Bigl( \frac{3(m+1)}{4m}\arcsin\gamma \Bigr) & m\text{ odd}
    \end{cases}
\end{equation}
\cref{eq: upper bound second method} is the second analytical upper bound to \eqref{eq: problem def for ms2t}. For our simulations, we used a tight concave upper-bound:
\[
g(\theta) = \begin{cases}
0.724611\cdot  \theta & 0\leq \theta\leq 1.16556 \\
\sin^2(\theta) & 1.16556 < \theta < \frac{\pi}{2} \\
    1 & \tfrac{\pi}{2} \leq \theta 
    \leq \pi.
\end{cases}
\]

\subsection{Combining the bounds}\label{sec: results}
In \cref{fig: sim results} we plot the lower bounds of \eqref{eq: analytical lb degenerate} and \eqref{eq: analytical lb uniform} and upper bounds of \eqref{eq: upper bound first method} and \eqref{eq: upper bound second method} normalized by $m^2$ for $m=3,6,9$ and $m\to\infty$. Different bounds are tighter for different values of $\gamma$ and $m$. The gap between our best upper and lower bounds (shaded region) is quite small for all $m$.

We also solved \eqref{eq: problem def for ms2t} directly with a local optimizer%
\footnote{\hlyellow{We used MATLAB's {\footnotesize\tt fmincon} local optimizer with the default {\footnotesize\tt interior-point} algorithm and initialized the decision variables with $\alpha_{i,k} \sim \mathrm{Uniform}[-\arcsin\gamma,\arcsin\gamma]$ and $\eta=0$ for each $\gamma,m,i,k$. Using the algorithm {\footnotesize\tt trust-region-reflective} yielded similar results, but {\footnotesize\tt sqp} did not work at all.}}. As seen in \cref{fig: sim results}, the optimizer is often trapped in local minima, leading to solutions that are worse than our lower bounds.
Moreover, the optimizer never beats our lower bound, suggesting that our lower bound may be globally optimal.
\begin{figure}[!ht]
    \centering
    \includegraphics{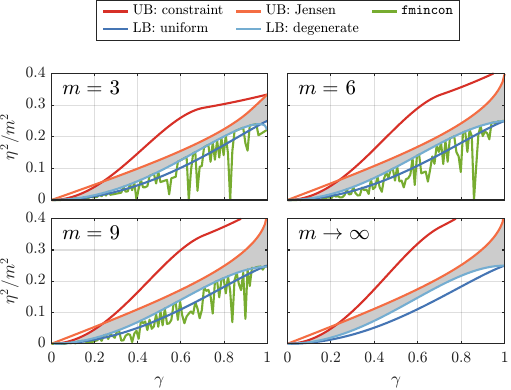}
    \caption{Normalized upper and lower bounds provided in \cref{subsec: lower bound analytical} and \cref{subsec: upper bound analytical} for $m=3,6,9$ and $m\to\infty$. In the first three simulations, the exact problem in \eqref{eq: problem def for ms2t} has been solved using {\footnotesize\tt fmincon} in MATLAB. The shaded region indicates the gap between best upper and lower bounds.}
    \label{fig: sim results}
\end{figure}

\section{CONCLUSION AND FUTURE RESEARCH}\label{sec: Conclusion and Future Work}
We considered the problem of stealthy optimal sensor placement, where a network of range sensors attempts to maximize the localization information obtained about a set of targets, while limiting the information revealed to the targets (which also have range sensors). We provided analytical solutions for $m=2$ sensors and $n=2$ targets. We also found upper and lower bounds for $m\geq 3$ and $n=2$.
One direction for future research is to explore different sensor models such as bearing-only sensors or sensors with range-dependent noise.
\hlcyan{Another area of exploration is the case of $n\geq 3$ targets.} In 
\cref{fig: feasible 4 targets}, we show the region satisfying a stealthiness constraint for a particular arrangement of $n=4$ targets. This is far more complex than the $n=2$ case (see \cref{fig: level sets}). Small changes in $\gamma$ or the target locations can cause dramatic changes to the topology and connectedness of the feasible set.\looseness=-1

\begin{figure}[!ht]
    \centering
\includegraphics{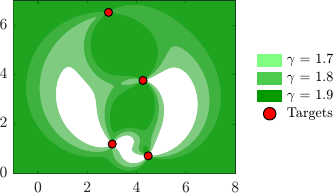}
    \caption{Feasible set corresponding to $\sum_{1\leq k < \ell \leq 4}\sin^2{\beta_{k\ell}} \leq \gamma^2$ for different $\gamma$ and $n=4$ targets. Compared to the $n=2$ case (see \cref{fig: level sets}), the feasible set is complicated; small changes in $\gamma$ cause dramatic changes in shape and connectedness.}
    \label{fig: feasible 4 targets}
\end{figure}

\section{ACKNOWLEDGMENTS} The authors wish to thank Dr. Christopher M. Kroninger (DEVCOM ARL) for support and insights on this topic.

\bibliographystyle{IEEEtran}
\bibliography{references}
\end{document}